\documentclass[11pt,letterpaper]{article}
\usepackage[lmargin=1.0in,rmargin=1.0in,bottom=1.0in,top=1.0in,twoside=False]{geometry}

\usepackage{fullpage,amssymb,amsmath}
\usepackage{graphicx}
\usepackage{tikz}
\usetikzlibrary{shapes}
\usetikzlibrary{calc}

\usepackage{skull}
\usepackage[T1]{fontenc}

 \usepackage{xcolor}
 \usepackage{mathtools}
\usepackage{microtype}
\usepackage{amsfonts}
\usepackage{comment}
\usepackage[english]{babel}
\usepackage{mathrsfs}
\usepackage[vlined, ruled, linesnumbered]{algorithm2e}
\DontPrintSemicolon

\usepackage{trimspaces}
\usepackage{nccfoots}
\usepackage{setspace}
\usepackage{inconsolata}
\usepackage{libertine}
\usepackage[absolute]{textpos}
\usepackage{thmtools}
\usepackage{thm-restate}

\usepackage{xspace}
\usepackage{enumitem}

\usepackage{todonotes}

\usepackage{longtable}

\definecolor{blue}{rgb}{0.1,0.2,0.5}
\definecolor{brown}{rgb}{0.6,0.6,0.2}
\usepackage[ocgcolorlinks, linkcolor={blue}, citecolor={brown}]{hyperref}

\usepackage[amsmath,thmmarks,hyperref]{ntheorem}
\usepackage{cleveref}

\crefformat{page}{#2page~#1#3}%
\Crefformat{page}{#2Page~#1#3}%
\crefformat{equation}{#2(#1)#3}%
\Crefformat{equation}{#2(#1)#3}%
\crefformat{figure}{#2Figure~#1#3}%
\Crefformat{figure}{#2Figure~#1#3}%
\crefformat{section}{#2Section~#1#3}
\Crefformat{section}{#2Section~#1#3}
\crefformat{chapter}{#2Chapter~#1#3}
\Crefformat{chapter}{#2Chapter~#1#3}
\crefformat{chapter*}{#2Chapter~#1#3}
\Crefformat{chapter*}{#2Chapter~#1#3}
\crefformat{part}{#2Part~#1#3}
\Crefformat{part}{#2Part~#1#3}
\crefformat{enumi}{#2(#1)#3}
\Crefformat{enumi}{#2(#1)#3}

\usepackage{latexsym}

\theoremnumbering{arabic}
\theoremstyle{plain}
\theoremsymbol{}
\theorembodyfont{\itshape}
\theoremheaderfont{\normalfont\bfseries}
\theoremseparator{.}

\newtheorem{theorem}{Theorem}
\crefformat{theorem}{#2Theorem~#1#3}
\Crefformat{theorem}{#2Theorem~#1#3}

\newcommand{\newtheoremwithcrefformat}[2]{%
  \newtheorem{#1}[theorem]{#2}%
  \crefformat{#1}{##2\MakeUppercase#1~##1##3}%
  \Crefformat{#1}{##2\MakeUppercase#1~##1##3}%
}
\newcommand{\newseptheoremwithcrefformat}[2]{%
  \newtheorem{#1}{#2}%
  \crefformat{#1}{##2\MakeUppercase#1~##1##3}%
  \Crefformat{#1}{##2\MakeUppercase#1~##1##3}%
}

\newtheoremwithcrefformat{lemma}{Lemma}
\newtheoremwithcrefformat{proposition}{Proposition}
\newtheoremwithcrefformat{observation}{Observation}
\newtheoremwithcrefformat{conjecture}{Conjecture}
\newtheoremwithcrefformat{corollary}{Corollary}
\newseptheoremwithcrefformat{claim}{Claim}
\theorembodyfont{\upshape}
\newtheoremwithcrefformat{example}{Example}
\newtheoremwithcrefformat{remark}{Remark}
\newseptheoremwithcrefformat{definition}{Definition}
\newseptheoremwithcrefformat{question}{Question}

\theoremstyle{nonumberplain}
\theoremheaderfont{\scshape}
\theorembodyfont{\normalfont}
\theoremsymbol{\ensuremath{\square}}
\newtheorem{proof}{Proof}

\theoremsymbol{\ensuremath{\lrcorner}}
\newtheorem{claimproof}{Proof of Claim}

\def\cqedsymbol{\ifmmode$\lrcorner$\else{\unskip\nobreak\hfil
\penalty50\hskip1em\null\nobreak\hfil$\lrcorner$
\parfillskip=0pt\finalhyphendemerits=0\endgraf}\fi}

\newcommand{\ZZ}{{\ensuremath{\mathbb{Z}}}}
\newcommand{\NN}{{\ensuremath{\mathbb{N}}}}
\newcommand{\Oh}{\ensuremath{\mathcal{O}}}

\newcommand{\NP}{$\mathsf{NP}$\xspace}
\newcommand{\XP}{$\mathsf{XP}$\xspace}
\newcommand{\FPT}{$\mathsf{FPT}$\xspace}

\newcommand{\ve}[1]{\ensuremath{\mathbf{#1}}}

\newcommand{\veb}{\ensuremath{\mathbf{b}}}

\newcommand{\vex}{\ensuremath{\mathbf{x}}}

\DeclareMathOperator{\td}{\operatorname{td}}

\DeclareMathOperator{\maxdeg}{\operatorname{maxdeg}}
\DeclareMathOperator{\assign}{\operatorname{assignment}}

\newcommand{\appendixText}{}
\newcommand{\toappendix}[1]{\gappto{\appendixText}{{#1}}}

\renewcommand{\leq}{\leqslant}

\renewcommand{\le}{\leqslant}

\begin{document}

\title{Integer Programming and Incidence Treedepth\thanks{The manuscript is an extended version of article~\cite{DBLP:conf/ipco/EibenGKOPW19}, which appeared in the proceedings of IPCO'19.
This work is a part of projects CUTACOMBS, PowAlgDO (M.~Wrochna) and TOTAL (M.~Pilipczuk) that
have received funding from the European Research Council (ERC) under the European Union’s
Horizon 2020 research and innovation programme (grant agreements No. 714704, No 714532, and No. 677651).
Robert Ganian is supported by the Austrian Science Fund (FWF Project P31336).
Marcin Wrochna is supported by Foundation for Polish Science (FNP) via the START stipend, and this work was partially done while he was affiliated with the Institute of Informatics, University of Warsaw, Poland.
}}

\newcommand{\email}[1]{\texttt{#1}}

\author{
Eduard Eiben\thanks{Department of Computer Science, Royal Holloway, University of London, UK, \email{eduard.eiben@rhul.ac.uk}}
\and
Robert Ganian\thanks{Algorithms and Complexity Group, Vienna University of Technology, Austria,
\email{rganian@ac.tuwien.ac.at}
}
\and
Dušan Knop\thanks{Department of Theoretical Computer Science, Czech Technical University in Prague, Czech Republic, \email{dusan.knop@fit.cvut.cz}
}
\and
Sebastian Ordyniak\thanks{School of Computing, University of Leeds, UK,
\email{s.ordyniak@gmail.com}}
\and
Micha\l~Pilipczuk\thanks{Institute of Informatics, University of Warsaw, Poland,
\email{michal.pilipczuk@mimuw.edu.pl}}
\and
Marcin Wrochna\thanks{University of Oxford, United Kingdom, \email{marcin.wrochna@mimuw.edu.pl}}
}

%
%
\maketitle              

\begin{textblock}{20}(13.9, 10.7)
\includegraphics[width=40px]{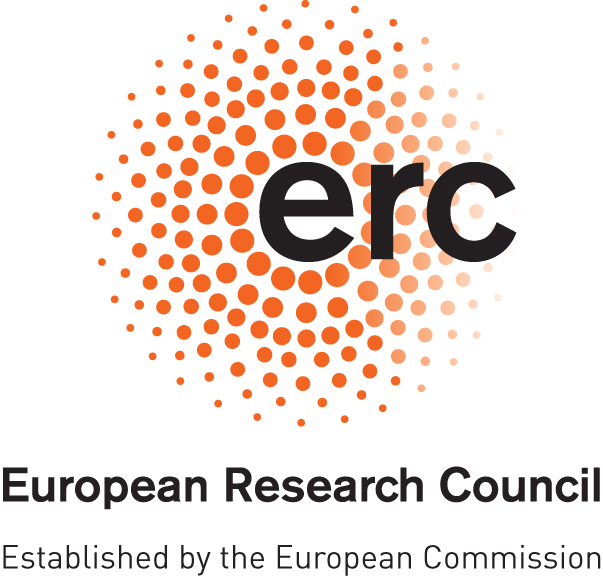}%
\end{textblock}
\begin{textblock}{20}(13.65, 11)
\includegraphics[width=60px]{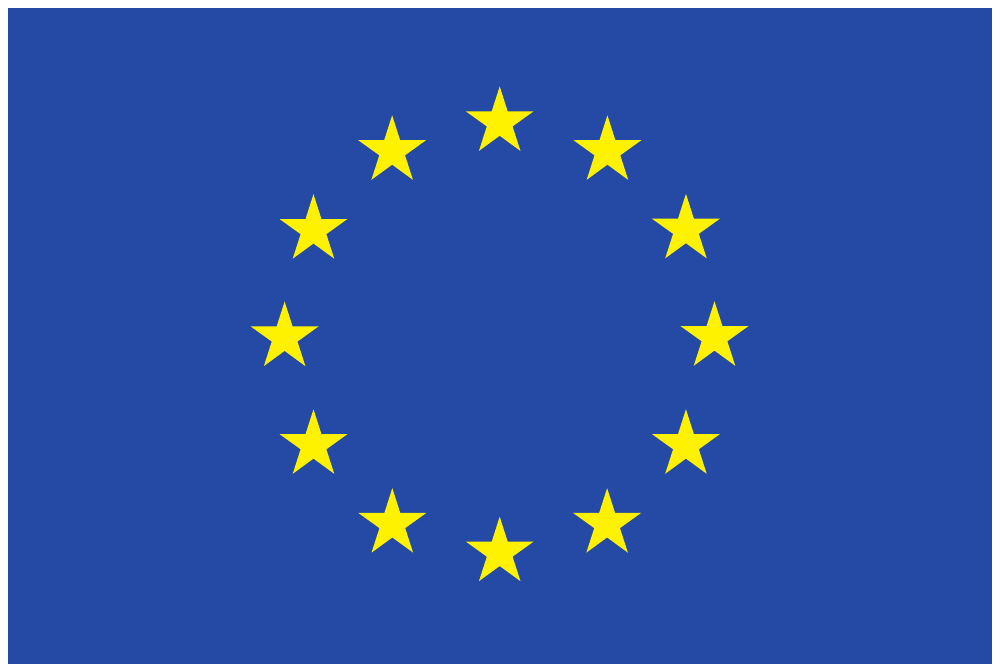}%
\end{textblock}

\begin{abstract}
Recently a strong connection has been shown between the trac\-ta\-bil\-i\-ty of integer programming (IP) with bounded coefficients on the one side and the structure  
of its constraint matrix on the other side.
To that end, integer linear programming is fixed-parameter tractable with respect to the primal (or dual) treedepth of the Gaifman graph of its constraint matrix and the largest coefficient (in absolute value).
Motivated by this, Koutecký, Levin, and Onn [ICALP 2018] asked whether it is possible to extend these result to a more broader class of integer linear programs.
More formally, is integer linear programming fixed-parameter tractable with respect to the incidence treedepth of its constraint matrix and the largest coefficient (in absolute value)?

We answer this question in negative.
In particular, we prove that deciding the feasibility of a system in the standard form, ${A\vex = \veb}, {\ve{l} \le \vex \le \ve{u}}$, is \NP-hard even when the absolute value of any coefficient in $A$ is 1 and the incidence treedepth of $A$ is 5.
Consequently, it is not possible to decide feasibility in polynomial time even if both the assumed parameters are constant, unless $\mathsf{P}=\mathsf{NP}$.
Moreover, we complement this intractability result by showing tractability for natural and only slightly more
restrictive settings, namely: (1) treedepth with an additional bound on either the maximum arity of constraints or the maximum number of occurrences of variables and (2) the vertex cover number.

\end{abstract}

%
%
%
%
%
\newcommand{\bigoh}{\mathcal{O}}

\section{Introduction}
In this paper we consider the decision version of Integer Linear Program (ILP) in \emph{standard form}.
Here, given a matrix $A \in \ZZ^{m \times n}$ with $m$ rows (constraints) and $n$ columns and vectors $\veb \in \ZZ^m$ and $\ve{l}, \ve{u} \in \ZZ^n$ the task is to decide whether the set
\begin{equation}
\left\{ \vex \in \ZZ^n \mid A\vex = \veb, \, \ve{l} \leq \vex \leq \ve{u} \right\}  \tag{SSol}\label{StandardLinIP}
\end{equation}
is non-empty.
We are going to study structural properties of the incidence graph of the matrix $A$.
An integer program (IP) is a \emph{standard IP} (SIP) if its set of solutions is described by \eqref{StandardLinIP}, that is, if it is of the form
\begin{equation}
\min \left\{ f(\vex) \mid A\vex = \veb, \, \ve{l} \leq \vex \leq \ve{u} \,, \vex \in \ZZ^n \right\} \,, \tag{SIP}\label{SIP}
\end{equation}
where $f\colon \NN^n \to \NN$ is the \emph{objective function}; in case $f$ is a linear function the above SIP is said to be a linear SIP.
Before we go into more details we first review some recent development concerning algorithms for solving (linear) SIPs in variable dimension with the matrix $A$ admitting a certain decomposition.

Let $E$ be a $2 \times 2$ block matrix, that is, $E = \left(\begin{smallmatrix} A_1 & A_2 \\ A_3 & A_4 \end{smallmatrix}\right)$, where $A_1, \ldots, A_4$ are integral matrices.
We define an \emph{$n$-fold 4-block product of $E$} for a positive integer $n$ as the following block matrix
\[
E^{(n)} =
\begin{pmatrix}
A_1 & A_2 & A_2 & \cdots & A_2 \\
A_3 & A_4 & 0 & \cdots & 0 \\
A_3 & 0 & A_4 & \cdots & 0 \\
\vdots &  &  & \ddots &  \\
A_3 & 0 & 0 & \cdots & A_4
\end{pmatrix},
\]
where $0$ is a matrix containing only zeros (of appropriate size).
One can ask whether replacing $A$ in the definition of the set of feasible solutions \eqref{StandardLinIP} can give us an algorithmic advantage leading to an efficient algorithm for solving such SIPs.
We call such an SIP an \emph{$n$-fold 4-block IP}.
We derive two special cases of the $n$-fold 4-block IP with respect to special cases for the matrix $E$ (see monographs~\cite{DeLoeraHK13,Onn10} for more information).
If both $A_1$ and $A_3$ are void (not present at all), then the result of replacing $A$ with $E^{(n)}$ in \eqref{SIP} yields the \emph{$n$-fold IP}.
Similarly, if $A_1$ and $A_2$ are void, we obtain the \emph{2-stage stochastic IP}.

The first, up to our knowledge, pioneering algorithmic work on $n$-fold 4-block IPs is due to Hemmecke et al.~\cite{HemmeckeKW10}.
They gave an algorithm that given $n$, the $2 \times 2$ block matrix $E$, and vectors $\ve{w}, \ve{b}, \ve{l}, \ve{u}$ finds an integral vector $\vex$ with $E^{(n)}\vex = \ve{b}, \ve{l} \le \vex \le \ve{u}$ minimizing $\ve{w}\vex$.
The algorithm of Hemmecke et al.~\cite{HemmeckeKW10} runs in time $n^{g(r,s,\| E \|_\infty)}L$, where $r$ is the number of rows of $E$, $s$ is the number of columns of $E$, $L$ is the size of the input, and $g\colon \NN \to \NN$ is a computable function.
Thus, from the parameterized complexity viewpoint this is an \XP algorithm for parameters $r,s, \| E \|_\infty$.
This algorithm has been recently improved by Chen et al.~\cite{ChenXS18} who give better bounds on the function $g$; it is worth noting that Chen et al.~\cite{ChenXS18} study also the special case where $A_1$ is a zero matrix and even in that case present an \XP algorithm.
Since the work of Hemmecke et al.~\cite{HemmeckeKW10}
the question of whether it is possible to improve the algorithm to run in time $g'(r,s,\|E\|_\infty) \cdot n^{\bigoh(1)}L$ or not has become a major open question in the area of mathematical programming.

Of course, the complexity of the two aforementioned special cases of $n$-fold 4-block IP are extensively studied as well.
The first \FPT algorithm\footnote{That is, an algorithm running in time $f(r,s,\|E\|_\infty) \cdot n^{\bigoh(1)}L$.} for the $n$-fold IPs (for parameters $r,s,\| E \|_\infty$) is due to Hemmecke et al.~\cite{HemmeckeOR13}.
Their algorithm has been subsequently improved~\cite{KouteckyLO18,EisenbrandHK18}.
Altmanová et al.~\cite{AltmanovaKK18} implemented the algorithm of Hemmecke et al.~\cite{HemmeckeOR13} and improved the polynomial factor (achieving the same running time as Eisenbrand et al.~\cite{EisenbrandHK18}) the above algorithms (from cubic dependence to $n^2 \log n$).
The best running time of an algorithm solving $n$-fold IP is due to Jansen et al.~\cite{JansenLR18} and runs in nearly linear time in terms of $n$.

Last but not least, there is an \FPT algorithm for solving the 2-stage stochastic IP due to Hemmecke and Schultz~\cite{HemmeckeS03}.
This algorithm is, however, based on a well quasi ordering argument yielding a bound on the size of the Graver basis for these IPs.
Very recently Klein~\cite{Klein19} presented a constructive approach using Steinitz lemma and give the first explicit (and seemingly optimal) bound on the size of the Graver basis for 2-stage (and multistage) IPs.
It is worth noting that possible applications of 2-stage stochastic IP are much less understood than those of its counterpart $n$-fold IP.

In the past few years, algorithmic research in this area has been mainly application-driven.
Substantial effort has been taken in order to find the right formalism that is easier to understand and yields algorithms having the best possible ratio between their generality and the achieved running time.
It turned out that the right formalism is connected with variants of the Gaifman graph (see e.g.~\cite{Dechter06}) of the matrix $A$ (for the definitions see the Preliminaries section).

\paragraph{Our Contribution.}
In this paper we focus on the incidence (Gaifman) graph.
We investigate the (negative) effect of the treedepth of the incidence Gaifman graph on tractability of ILP feasibility.

\begin{theorem}\label{thm:incidenceTD}
Given a matrix $A\in \{-1,0,1\}^{m \times n}$ and vectors $\ve{l}, \ve{u}\in \ZZ_\infty^n$.
Deciding whether the set defined by \eqref{StandardLinIP} is non-empty is \NP{}-hard even if $\veb = \mathbf{0}$ and $\td_I(A) \le 5$.
\end{theorem}
We complement \cref{thm:incidenceTD} (\cref{sec:trac}),
by showing that \eqref{SIP} becomes fixed-parameter tractable
parameterized by $\|A\|_\infty$ and either:
\begin{itemize}[nosep]
\item treedepth plus $\min\{\maxdeg_C(A),\maxdeg_V(A)\}$, where
  $\maxdeg_C(A)$ is the maximum arity of any constraint and
  $\maxdeg_V(A)$ is the maximum number of constraints any
  variable occurs~in,
\item the vertex cover number of $G_I(A)$.
\end{itemize}

\subsection*{Preliminaries}
For integers $m < n$ by $[m:n]$ we denote the set $\left\{ m, m+1, \ldots, n \right\}$ and $[n]$ is a shorthand for $[1:n]$.
We use bold face letters for vectors and normal font when referring to their components, that is, $\vex$ is a vector and $x_3$ is its third component.
For vectors of vectors we first use superscripts to access the ``inner vectors'', that is, $\vex = (\vex^1, \ldots, \vex^n)$ is a vector of vectors and $\vex^3$ is the third vector in this collection.

\paragraph{From Matrices to Graphs.}
Let $A$ be an $m \times n$ integer matrix.
The \emph{incidence Gaifman graph} of $A$ is the bipartite graph $G_I = (R \cup C, E)$, where $R = \left\{ r_1, \ldots, r_m \right\}$ contains one vertex for each row of $A$ and $C = \left\{ c_1, \ldots, c_n \right\}$ contains one vertex for each column of $A$.
There is an edge $\{ r, c\}$ between the vertex $r \in R$ and $c \in C$ if $A(r,c) \neq 0$, that is, if row $r$ contains a nonzero coefficient in column $c$.
The \emph{primal Gaifman graph} of $A$ is the graph $G_P = (C, E)$, where $C$ is the set of columns of $A$ and $\{ c, c' \} \in E$ whenever there exists a row of $A$ with a nonzero coefficient in both columns $c$ and $c'$.
The \emph{dual Gaifman graph} of $A$ is the graph $G_D = (R, E)$, where $R$ is the set of rows of $A$ and $\{ r, r' \} \in E$ whenever there exists a column of $A$ with a nonzero coefficient in both rows $r$ and $r'$.

\paragraph{Treedepth.}
Undoubtedly, the most celebrated structural parameter for graphs is treewidth, however, in the case of ILPs bounding treewidth of any of the graphs defined above does not lead to tractability (even if the largest coefficient in $A$ is bounded as well see e.g.~\cite[Lemma~18]{KouteckyLO18}).
Treedepth is a structural parameter which is useful in the theory of so-called sparse graph classes, see e.g.~\cite{NesetrilOdM:Sprasity}. Let $G = (V,E)$ be a graph.
The treedepth of $G$, denoted $\td(G)$, is defined by the following recursive formula:
\[
\td(G) = \begin{cases}
	1 &\mbox{if } |V(G)| = 1, \\[0.2cm]
	1+\min_{v\in V(G)}\td(G - v) &\mbox{if $G$ is connected with } |V(G)| > 1, \\[0.2cm]
	\max_{i \in [k]} \td(G_i) & \text{\begin{minipage}{.5\textwidth}if $G_1, \ldots, G_k$ are connected components of~$G$.\end{minipage}}
\end{cases}
\]
Let $A$ be an $m \times n$ integer matrix.
The \emph{incidence treedepth} of $A$, denoted $\td_I(A)$, is the treedepth of its incidence Gaifman graph $G_I$.
The \emph{dual treedepth} of $A$, denoted $\td_D(A)$, is the treedepth of its dual Gaifman graph $G_D$.
The \emph{primal treedepth} is defined similarly.

The following two well-known theorems will be used in the proof of \cref{thm:incidenceTD}.
\begin{theorem}[Chinese Remainder Theorem]
	Let $p_1,\ldots,p_n$ be pairwise co-prime integers greater than~$1$ and let $a_1, \ldots, a_n$ be integers such that for all $i\in [n]$ it holds $0\le a_i<p_i$. Then there exists exactly one integer $x$ such that
	\begin{enumerate}[nosep]
		\item $0\le x < \prod_{i = 1}^{n}p_i$ and
		\item $\forall i\in[n]\,\colon\, x\equiv a_i\mod p_i.$
	\end{enumerate}
\end{theorem}

\begin{theorem}[Prime Number Theorem]
	Let $\pi(n)$ denote the number of primes in $[n]$, then $\pi(n)\in\Theta(\frac{n}{\log n})$.
\end{theorem}

It is worth pointing out that, given a positive integer $n$ encoded in unary, it is possible to the $n$-th prime in polynomial time.

\section{Proof of \cref{thm:incidenceTD}}
Before we proceed to the proof of \cref{thm:incidenceTD} we include a brief sketch of its idea.
To prove \NP-hardness, we will give a polynomial time reduction from \textsc{3-SAT} which is well known to be \NP-complete~\cite{GareyJ79}.
The proof is inspired by the \NP-hardness proof for ILPs given by a set of inequalities, where the primal graph is a star, of Eiben et. al~\cite{EibenGKO18}.

\paragraph{Proof Idea.}
Let $\varphi$ be a 3-CNF formula.
We encode an assignment into a variable~$y$.
With every variable $v_i$ of the formula $\varphi$ we associate a prime number $p_i$.
We make \mbox{$y \bmod p_i$} be the boolean value of the variable $v_i$; i.e., using auxiliary gadgets we force $y\bmod p_i$ to always be in $\{0,1\}$.
Further, if for a clause $C \in \varphi$ by $\| C \|$ we denote the product of all of the primes associated with the variables occurring in $C$,
then, by Chinese Remainder Theorem, there is a single value in $[ \| C \| ]$, associated with the assignment that falsifies $C$, which we have to forbid for $y \bmod \| C \|$.
We use the box constraints, i.e., the vectors $\ve{l}, \ve{u}$, for an auxiliary variable taking the value $y \bmod \| C \|$ to achieve this.
For example let $\varphi = (v_1 \lor \neg v_2 \lor v_3)$ and let the primes associated with the three variables be $2,3,$ and $5$, respectively.
Then we have $\| (v_1 \lor \neg v_2 \lor v_3) \| = 30$ and, since $v_1 = v_3 = \texttt{false}$ and $v_2 = \texttt{true}$ is the only assignment falsifying this clause, we have that $21$ is the forbidden value for $y \bmod 30$.
Finally, the \eqref{SIP} constructed from $\varphi$ is feasible if and only if there is a satisfying assignment for $\varphi$.

\begin{proof}[of \cref{thm:incidenceTD}]
Let $\varphi$ be a 3-CNF formula with $n'$ variables $v_1,\ldots, v_{n'}$ and $m'$ clauses $C_1,\ldots, C_{m'}$ (an instance of \textsc{3-SAT}).
Note that we can assume that none of the clauses in $\varphi$ contains a variable along with its negation.
We will define an SIP, that is, vectors $\veb, \ve{l}, \ve{u}$, and a matrix $A$ with $\Oh((n'+m')^5)$ rows and columns, whose solution set is non-empty if and only if a satisfying assignment exists for $\varphi$.
Furthermore, we present a decomposition of the incidence graph of the constructed SIP proving that its treedepth is at most~5.
We naturally split the vector $\vex$ of the SIP into subvectors associated with the sought satisfying assignment, variables, and clauses of $\varphi$, that is, we have $\vex = \left( y, \vex^1, \ldots, \vex^{n'}, \ve{z}^1, \ldots, \ve{z}^{m'}\right)$.
Throughout the proof $p_i$ denotes the $i$-th prime number.

\paragraph{Variable Gadget.}
We associate the $\vex^i = \left( x^i_0, \ldots, x^i_{p_i} \right)$ part of $\vex$ with the variable $v_i$ and bind the assignment of $v_i$ to $y$.
We add the following constraints
\begin{align}
x^i_1 &= x^i_\ell & \forall \ell \in \left[2 : p_i \right]  \label{eq:x_is_equal} \\
x^i_0 &= y+\sum_{\ell=1}^{p_i}x^i_\ell &										\label{eq:x_is_mod}
\end{align}
and box constraints
\begin{align}
-\infty &\le x^i_\ell \le \infty & \forall \ell \in [p_i]  	\label{eq:x_unbounded_box} \\
0 &\le x^i_0 \le 1 &																				\label{eq:x_truth_box}
\end{align}
to the SIP constructed so far.
\begin{claim}\label{clm:variables}
For given values of $x^i_0$ and $y$, one may choose the values of $x^i_\ell$ for $\ell\in [p_i]$ so that \eqref{eq:x_is_equal} and \eqref{eq:x_is_mod} are satisfied if and only if
$x^i_0 \equiv y \mod p_i$.
\end{claim}
\begin{claimproof}
By \eqref{eq:x_is_equal} we know $x^i_1 = \cdots = x^i_{p_i}$ and thus by substitution we get the following equivalent form of \eqref{eq:x_is_mod}
\begin{equation}\label{eq:xi0andTheLargeCoefficient}
x^i_0 = y + p_i \cdot x^i_1 \,.
\end{equation}
But this form is equivalent to $x^i_0 \equiv y \mod p_i$.
\end{claimproof}
Note that by (the proof of) the above claim the conditions \eqref{eq:x_is_equal} and \eqref{eq:x_is_mod} essentially replace the large coefficient ($p_i$) used in the condition \eqref{eq:xi0andTheLargeCoefficient}.
This is an efficient trade-off between large coefficients and incidence treedepth which we are going to exploit once more when designing the clause gadget.

By the above claim we get an immediate correspondence between $y$ and truth assignments for $v_1, \ldots, v_{n'}$.
For an integer $w$ and a variable $v_i$ we define the following mapping
\[
\assign(w, v_i) =
\begin{cases}
\texttt{true} & \textrm{if } w \equiv 1 \mod p_i \\
\texttt{false} & \textrm{if } w \equiv 0 \mod p_i \\
\textrm{undefined} & \textrm{otherwise}.
\end{cases}
\]
Notice that \eqref{eq:x_truth_box} implies that the mapping $\assign(y, v_i) \in \{ \texttt{true}, \texttt{false} \}$ for $i \in \left[ n' \right]$.
We straightforwardly extend the mapping $\assign(\cdot, \cdot)$ for tuples of variables as follows.
For a tuple $\ve{a}$ of length $\ell$, the value of $\assign(w,\ve{a})$ is $(\assign(w,a_1), \ldots, \assign(w,a_\ell))$ and we say that $\assign(w,\ve{a})$ is defined if all of its components are defined.

\paragraph{Clause Gadget.}
Let $C_j$ be a clause with variables $v_e,v_f,v_g$.
We define $\| C_j \|$ as the product of the primes associated with the variables occurring in $C_j$, that is, $\| C_j \| = p_e \cdot p_f \cdot p_g$.
We associate the $\ve{z}^j = \left( z^j_0, \ldots, z^j_{\| C_j \|} \right)$ part of $\vex$ with the clause $C_j$.
Let $d_j$ be the unique integer in $[\| C_j \|]$ for which $\assign(d_j, (v_e,v_f,v_g))$ is defined and gives the falsifying assignment for $C_j$.
The existence and uniqueness of $d_j$ follows directly from the Chinese Remainder Theorem.
We add the following constraints
\begin{align}
z^j_1 &= z^j_\ell & \forall \ell \in \left[2 : \| C_j \| \right]  \label{eq:z_is_equal} \\
z^j_0 &= y+\sum_{1 \le \ell \le \| C_j \|} z^j_\ell &							\label{eq:z_is_mod}
\end{align}
and box constraints
\begin{align}
-\infty &\le z^j_\ell \le \infty & \forall \ell \in [\| C_j \|]  	\label{eq:z_unbounded_box} \\
 d_j + 1 &\le z^j_0 \le \| C_j \| + d_j - 1 &											\label{eq:z_truth_box}
\end{align}
to the SIP constructed so far.
\begin{claim}\label{clm:clauses}
Let $C_j$ be a clause in $\varphi$ with variables $v_e, v_f, v_g$.
For given values of $y$ and $z^j_0$ such that the value $\assign(y,(v_e, v_f, v_g))$ is defined, one may choose the values of $z^j_\ell$ for $\ell\in [\|C_j\|]$ so that \eqref{eq:z_is_equal}, \eqref{eq:z_is_mod},
\eqref{eq:z_unbounded_box} and \eqref{eq:z_truth_box} are satisfied if and only if $\assign(y,(v_e, v_f, v_g))$ satisfies $C_j$.
\end{claim}
\begin{claimproof}
Similarly to the proof of the \cref{clm:variables}, \eqref{eq:z_is_equal} and \eqref{eq:z_is_mod} together are equivalent to $z^j_0 \equiv y \mod \| C_j \|$.
Finally, by \eqref{eq:z_truth_box} we obtain that $z^j_0 \neq d_j$ which holds if and only if $\assign(y,(v_e, v_f, v_g))$ satisfies $C_j$.
\end{claimproof}

Let $A\vex = \ve{0}$ be the SIP with constraints \eqref{eq:x_is_equal}, \eqref{eq:x_is_mod}, \eqref{eq:z_is_equal}, and \eqref{eq:z_is_mod}
and box constraints $\ve{l} \le \vex \le \ve{u}$ given by \eqref{eq:x_unbounded_box}, \eqref{eq:x_truth_box}, \eqref{eq:z_unbounded_box}, \eqref{eq:z_truth_box}, and $-\infty \le y \le \infty$.
By the \cref{clm:variables}, constraints \eqref{eq:x_is_equal}, \eqref{eq:x_is_mod}, \eqref{eq:x_unbounded_box}, \eqref{eq:x_truth_box},
are equivalent to the assertion that $\assign(y, (v_1, \ldots, v_{n'}))$ is defined.
Then by the \cref{clm:clauses}, constraints \eqref{eq:z_is_equal}, \eqref{eq:z_is_mod}, \eqref{eq:z_unbounded_box}, \eqref{eq:z_truth_box}
are equivalent to checking that every clause in $\varphi$ is satisfied by $\assign(y, (v_1, \ldots, v_{n'}))$.
This finishes the reduction and the proof of its correctness.


In order to finish the proof we have to bound the number of variables and constraints in the presented SIP and to bound the incidence treedepth of $A$.
It follows from the Prime Number Theorem that $p_i=\Oh(i\log i)$.
Hence, the number of rows and columns of $A$ is at most $(n'+m')p_{n'}^3=\Oh((n'+m')^5)$.
\begin{claim}\label{clm:incidence_treedepth}
It holds that $\td_I(A) \le 5$.
\end{claim}
\begin{claimproof}
Let $G$ be the incidence graph of the matrix $A$.
It is easy to verify that $y$ is a cut-vertex in $G$.
Observe that each component of $G - y$ is now either a variable gadget for $v_i$ with $i \in [n']$ (we call such a component a \emph{variable component}) or a clause gadget for $C_j$ with $j \in [m']$ (we call such a component a \emph{clause component}).
Let $G_v^i$ be the variable component (of $G - y$) containing variables $\vex^i$ and $G_c^j$ be the clause component containing variables $\ve{z}^j$.
Let $t_v = \max_{\ell \in [n']} \td(G^\ell_v)$ and $t_c = \max_{\ell \in [m']} \td(G^\ell_c)$.
It follows that $\td(G) \le 1 + \max( t_v, t_c)$.

Refer to \cref{fig:variableGadget}.
Observe that if we delete the variable $x^i_1$ together with the constraint \eqref{eq:x_is_mod} from $G_v^i$, then each component in the resulting graph contains at most two vertices.
Each of these components contains either
\begin{itemize}[nosep]
\item
a variable $x^i_\ell$ and an appropriate constraint \eqref{eq:x_is_equal} (the one containing $x^i_\ell$ and~$x^i_0$) for some $\ell \in [2 \,:\, p_i]$ or
\item
the variable $x^i_0$.
\end{itemize}
Since treedepth of an edge is 2 and treedepth of the one vertex graph is 1, we have that $t_v \le 4$.
\begin{figure}[bt]
	\begin{minipage}{.6\textwidth}

\begin{tikzpicture}[node distance=1.2cm]
\tikzstyle{deletionSet}=[fill=gray!40]
\tikzstyle{eqn}=[draw]
\tikzstyle{var}=[draw,circle,inner sep=2pt]

\node[var,deletionSet] (y) {$y$};
\node[eqn,below of=y,deletionSet] (y-xi) {$x^i_0 = y + \sum_{\ell = 1}^{p_i} x^i_\ell$};

\begin{scope}[shift={(-1.6cm,-2.5cm)},node distance=1.6cm]  
\node[var] (xi0) {$x^i_0$};
\node[var,right of=xi0,deletionSet] (xi1) {$x^i_1$};
\end{scope}

\begin{scope}[shift={(-1.4cm,-3.8cm)}, node distance=3cm] 
\node[eqn] (xi1xi2) {$x^i_1 = x^i_2$};
\node[eqn,right of=xi1xi2] (xi1xipi) {$x^i_1 = x^i_2$};
\node at ($(xi1xi2)!.5!(xi1xipi)$) {$\cdots$};
\end{scope}

\node[var,below of=xi1xi2] (xi2) {$x^i_2$};
\node[var,below of=xi1xipi] (xipi) {$x^i_{p_i}$};
\node at ($(xi2)!.5!(xipi)$) {$\cdots$};

\draw (y) to (y-xi);
\draw (y-xi) to (xi0);
\draw (y-xi) to (xi1);
\draw (xi1) to (xi1xi2);
\draw (xi1) to (xi1xipi);
\draw (xi1xi2) to (xi2);
\draw (xi1xipi) to (xipi);

\draw (y-xi) to[out=180,in=190] (xi2);
\draw (y-xi) to[out=0,in=0] (xipi);
\end{tikzpicture}
	\end{minipage}
	\begin{minipage}{.4\textwidth}
	\caption{\label{fig:variableGadget}%
	The variable gadget for $u_i$ of 3-SAT instance together with the global variable $y$.
	Variables (of the IP) are in circular nodes while equations are in rectangular ones.
	The nodes deleted in the proof of \cref{clm:incidence_treedepth} have light gray background.
	}
	\end{minipage}
\end{figure}
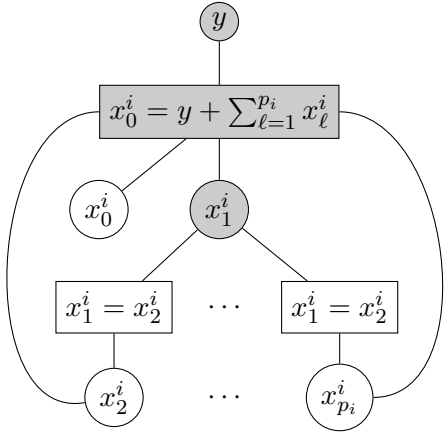

The bound on $t_c$ follows the same lines as for $t_v$, since indeed the two gadgets have the same structure.
Now, after deleting $z^j_1$ and \eqref{eq:z_is_mod} in $G_c^j$ we arrive to a graph with treedepth of all of its components again bounded by two (in fact, none of its components contain more than two vertices).
Thus, $t_v \le 4$ and the claim follows.
\end{claimproof}
The theorem follows by combining \cref{clm:variables}, \cref{clm:clauses}, and \cref{clm:incidence_treedepth}.
\end{proof}

\section{Complementary Tractability Results}
\label{sec:trac}

\paragraph{Treedepth and Degree Restrictions} It is worth noting that the proof of \cref{thm:incidenceTD} crucially
relies on having variables as well as constraints which have high
degree in the incidence graph.
Thus, it is natural to ask whether this is necessary or, equivalently, whether bounding the degree of variables, constraints, or both leads to tractability.
It is well known that if a graph $G$ has bounded degree and treedepth, then it is of bounded size, since indeed the underlying decomposition tree has bounded height and degree and thus bounded number of vertices.
Let \eqref{SIP} with $n$ variables be given.
Let $\maxdeg_C(A)$ denote the maximum arity of a constraint in its constraint matrix $A$ and let $\maxdeg_V(A)$ denote the maximum occurrence of a variable in constraints of $A$.
In other words, $\maxdeg_C(A)$ denotes the maximum number of nonzeros in a row of $A$ and $\maxdeg_V(A)$ denotes the maximum number of nonzeros in a column of $A$.
Now, we get that ILP can be solved in time $f(\maxdeg_C(A), \maxdeg_V(A), \td_I(A))L^{\bigoh(1)}$, where $f$ is some computable function and $L$ is the length of the encoding of the given ILP thanks to Lenstra's algorithm~\cite{Lenstra83}.

The above observation can in fact be
strengthened---namely, if the arity of all the constraints or the number of occurences of all the variables in the given SIP is bounded, then we obtain a bound on either primal or dual treedepth.
This is formalized by the following lemma.
\begin{lemma}\label{lem:treedepthDeltaBound}
For every \eqref{SIP} we have $$\td_P(A) \le \maxdeg_C(A) \cdot \td_I(A)\qquad\textrm{and}\qquad\td_D(A) \le \maxdeg_V(A) \cdot \td_I(A).$$
\end{lemma}

\begin{proof}
The proof idea is to investigate the definition of the incidence treedepth of $A$,
which essentially boils down to recursively eliminating either a row, or a column, or decomposing a block-decomposable matrix into its blocks.
Then, say for the second inequality above, eliminating a column can be replaced by eliminating all the at most $\maxdeg_V(A)$ rows that contain non-zero entries in this column.

We now proceed to the proof itself---in particular, we prove only the second inequality, as the first one is completely symmetric.
The proof is uses induction with respect to the total number of rows and columns of the matrix $A$.
The base of the induction, when $A$ has one row and one column, is trivial, so we proceed to the induction step.

Observe that $G_I(A)$ is disconnected if and only if $G_D(A)$ is disconnected if and only if $A$ is a block-decomposable matrix.
Moreover, the incidence treedepth of $A$ is the maximum incidence treedepth among the blocks of $A$, and the same also holds for the dual treedepth.
Hence, in this case we may apply the induction hypothesis to every block of $A$ and combine the results in a straightforward manner.

Assume then that $G_I(A)$ is connected. Then
$$\td(G_I(A))=1+\min_{v\in V(G_I(A))} \td(G_I(A)-v).$$
Let $v$ be the vertex for which the minimum on the right hand side is attained.
We consider two cases: either $v$ is a row of $A$ or a column of $A$.

Suppose first that $v$ is a row of $A$. Then we have
\begin{eqnarray*}
\td(G_D(A)) & \leq & 1+\td(G_D(A)-v)\\
            & \leq & 1+\maxdeg_V(A)\cdot \td(G_I(A)-v)\\
            & =    & 1+\maxdeg_V(A)\cdot (\td(G_I(A))-1)\\
            & \leq & \maxdeg_V(A)\cdot \td(G_I(A))
\end{eqnarray*}
as required, where the second inequality follows from applying the induction assumption to $A$ with the row $v$ removed.

Finally, suppose that $v$ is a column of $A$. Let $X$ be the set of rows of $A$ that contain non-zero entries in column $v$; then $|X|\leq \maxdeg_V(A)$ and $X$ is non-empty, because $G_I(A)$ is connected.
If we denote by $A-v$ the matrix obtained from $A$ by removing column $v$, then we have
\begin{eqnarray*}
\td(G_D(A)) & \leq & |X|+\td(G_D(A)-X) \\
            & \leq & \maxdeg_V(A)+\td(G_D(A-v))\\
            & \leq & \maxdeg_V(A)+\maxdeg_V(A)\cdot \td(G_I(A-v))\\
            & \leq & \maxdeg_V(A)\cdot \td(G_I(A)),
\end{eqnarray*}
as required. Here, in the second inequality we used the fact that $G_D(A)-X$ is a subgraph of $G_D(A-v)$, while in the third inequality we used the induction assumption for the matrix $A-v$.
\end{proof}

It follows that if we bound either $\maxdeg_V(A)$ or $\maxdeg_C(A)$, that is, formally set $\maxdeg(A) = \min \left\{ \maxdeg_V(A), \maxdeg_C(A) \right\}$,
then we can use the results of Koutecký et al.~\cite{KouteckyLO18} to solve the linear IP with such a solution set in time $f(\maxdeg(A), \|A\|_\infty) \cdot n^{\bigoh(1)} \cdot L$.
Consequently, the use of high-degree constraints and variables in the proof of \cref{thm:incidenceTD} is unavoidable.

\paragraph{Vertex Cover Number} It is natural to ask, whether there are other (more restrictive) structural
parameters than treedepth that allow for polynomial-time or even
fixed-parameter tractability for \eqref{SIP}. Indeed, one such
parameter is the (mixed) fracture number of $G_I(A)$, which was
introduced in~\cite{DvorakEGKO17} and is defined as the minimum
integer $k$ such that $G_I(A)$ has a deletion set $D$ of size at most
$k$ ensuring that every component of $G_I(A)\setminus D$ has size at most
$k$. It is easy to see that the treedepth of a graph is upper bounded by twice its fracture number.
It has been shown in~\cite[Corollary
8]{DvorakEGKO17} that \eqref{SIP} becomes solvable in polynomial-time
if both the fracture number of $G_I(A)$ and $\|A\|_\infty$ are bounded
by a constant. Moreover the question whether this result can be improved to
fixed-parameter tractability is known to be equivalent to the corresponding and
long-standing open questions for 4-block $n$-fold
ILPs~\cite{DvorakEGKO17}. Though we are not able to resolve this
question, we can at least show fixed-parameter tractability for a slightly more
restrictive parameter than fracture number, namely, the vertex cover
number of $G_I(A)$. Towards this result, we need the following
auxiliary corollary, which follows easily from~\cite[Theorem
4.1]{EisenbrandW18} and shows that \eqref{SIP} is fixed-parameter
tractable parameterized by both the number of rows $m$ in $A$ and $\|A\|_\infty$.
\begin{corollary}\label{cor:bounded-constraints}\sloppypar
  \eqref{SIP} can be solved in time
  $n\cdot \bigoh(m)^{m+3}\cdot \bigoh(\|A\|_\infty)^{m(m+1)}\cdot \log(m\|A\|_\infty)^2$, where $m$
  and $n$ are the number of rows respectively columns of $A$.
\end{corollary}
\begin{proof}
  Eisenbrand and Weismantel recently proved that
  the corollary holds if all variables in the given \eqref{SIP} have a
  lower bound of $0$, see~\cite[Theorem 4.1]{EisenbrandW18}. Since one can transform any \eqref{SIP} into an
  \eqref{SIP}, where all variables have a lower bound of $0$, by
  replacing any variable $x_i$, where $l_i\neq 0$, with
  $x_i'+l_i$, where $x_i'$ is a new variable with bounds $0 \leq
  x_i'\leq u_i-l_i$, and subtracting $A_{*i}l_i$ (where $A_{*i}$ denotes the
  $i$-th column of~$A$)  from $b$, we obtain
  that~\cite[Theorem 4.1]{EisenbrandW18} holds for general \eqref{SIP}.
\end{proof}

\begin{theorem}\sloppypar
  \eqref{SIP} can be solved in time
  $n\cdot \bigoh(k)^{2k+3}\cdot \bigoh(\|A\|_\infty)^{k(4k+2)}\cdot \log(2k\|A\|_\infty)^2$, where
  $k$ is the size of a minimum vertex cover for $G_I(A)$.
\end{theorem}
\begin{proof}
  Let $I$ be an instance of \eqref{SIP} with matrix $A$. 
  It is well-known, see e.g.~\cite[Chapter~1]{platypus}, that a minimum vertex cover of an $n$-vertex graph can be found in time $2^k\cdot \Oh(kn)$, where $k$ is its size. 
  Hence, we may assume that we are given
  a vertex cover $C$ of $G_I(A)$ of size $k$. Let $O_C$ be the set of all
  constraints that correspond to vertices in $G_I(A)\setminus
  C$. Because $C$ is a vertex cover, we obtain that the constraints in
  $O_C$ can only contain the at most $k$ variables in $C$. Moreover,
  since we can assume that all rows of $A$ are linear independent,
  we obtain that $|O_C|\leq k$. Hence $m\leq 2k$ and the theorem now
  follows from \cref{cor:bounded-constraints}.
\end{proof}

\section{Conclusions}

We have shown that, unlike the primal and the dual treedepth, the incidence treedepth of a constraint matrix of \eqref{SIP} does not (together with the largest coefficient) provide a way to tractability.
This shows that our current understanding of the structure of the incidence Gaifman graph is not sufficient.
Furthermore, it is not hard to see that the matrix $A$ in our hardness result (cf. \cref{fig:variableGadget}) has topological length $3$ (and height~$4$).
Topological length is a newly introduced parameter (\cite[Definition~18]{EisenbrandHKKLO19}) that allows to contract vertices of degree two in the tree witnessing bounded treedepth (i.e., in the tree in whose closure the incidence Gaifman graph emerges as a subgraph).
It is worth pointing out that in our reduction we have topological height $3$ and constant height while the $N$-fold $4$-block IP structure implies topological height $2$ and the height of the two levels is an additional parameter.
This further stimulates the question of whether an \FPT algorithm for $N$-fold $4$-block IP exists or not.
Thus, the effect on tractability of some other ``classical'' graph parameters shall be investigated.

Namely, whether ILP parameterized by the largest coefficient and treewidth and the maximum degree of the incidence Gaifman graph is in \FPT or not.
Furthermore, one may also ask about parameterization by the largest coefficient and the feedback vertex number of the incidence Gaifman graph.


%
\bibliographystyle{abbrv}
\bibliography{treedepth}


\end{document}